\documentclass[11pt]{article}

\usepackage[margin=1in]{geometry}
\usepackage{amsthm}
\usepackage{amssymb}
\usepackage{amsmath}
\usepackage{mathrsfs}
\usepackage{graphicx}
\usepackage{url}
\usepackage{color}
\usepackage{framed}

\newcommand\blfootnote[1]{%
  \begingroup
  \renewcommand\thefootnote{}\footnote{#1}%
  \addtocounter{footnote}{-1}%
  \endgroup
}

\newtheorem{theorem}{Theorem}

\newtheorem{lemma}[theorem]{Lemma}
\newtheorem{proposition}[theorem]{Proposition}

\theoremstyle{definition}

\newtheorem{definition}[theorem]{Definition}

\title{Fair redistricting is hard}

\author{Richard Kueng\footnote{Institute for Quantum Information and Matter, California Institute of Technology, Pasadena, CA}\qquad Dustin~G.~Mixon\footnote{Department of Mathematics, The Ohio State University, Columbus, OH}\qquad Soledad Villar\footnote{Courant Institute of Mathematical Sciences and Center for Data Science, New York University, New York, NY}}
\date{}

\begin{document}
\maketitle

\blfootnote{Send correspondence to \texttt{mixon.23@osu.edu}}

\begin{abstract}
Gerrymandering is a long-standing issue within the U.S.\ political system, and it has received scrutiny recently by the U.S.\ Supreme Court.
In this note, we prove that deciding whether there exists a fair redistricting among legal maps is $\mathsf{NP}$-hard.
To make this precise, we use simplified notions of ``legal'' and ``fair'' that account for desirable traits such as geographic compactness of districts and sufficient representation of voters.
The proof of our result is inspired by the work of Mahanjan, Minbhorkar and Varadarajan that proves that planar $k$-means is $\mathsf{NP}$-hard.
\end{abstract}

\section{Introduction}
Gerrymandering is the manipulation of district boundaries in order to favor a party or class.
It has been an issue in the U.S.\ political system for centuries.
Over this period, several legal constraints, such as the Voting Rights Act of 1965, have been installed in order to avoid certain forms of gerrymandering.
Even so, map makers with an agenda can optimize their objective (e.g., maximize the number of seats for a certain party) subject to these legal constraints.
Today, gerrymandered maps are continually brought to court, where they are struck down as illegal under federal or state law with some regularity.

In the recent Supreme Court case \textit{Gill v.\ Whitford}, it was argued that the Wisconsin State Assembly map exhibits partisan gerrymandering.
The plaintiffs offered evidence based on the efficiency gap metric~\cite{efficiency, duchin} and partisan bias metrics~\cite{GelmanK:94,GrofmanK:07}.
The Supreme Court did not rule on whether these metrics must be part of a test for partisan gerrymandering, thereby leaving open a fundamental question: 
What is an appropriate metric to detect partisan gerrymandering?
The answer to this question will undoubtedly require a robust theoretical framework for analyzing gerrymandering, and fortunately, gerrymandering is currently an active research area.

Implicit in the notion of gerrymandering is the notion of a fair map, which is subjective, and therefore hard to define.
One interesting approach along these lines is to consider a probability distribution on all legal maps, and then say that a map is fair if it is a typical instance of such a distribution~\cite{ravier, fifield}.
In order to decide whether a map is typical, researchers consider relevant observables of the distribution and determine whether they land within typical values~\cite{mattingly_NC, mattingly_WI, pegden}.
Observables of interest include the proportion of seats to votes and the difference between the mean and median votes in different districts.

One pitfall of such an approach is that the set of legal maps appears to be computationally intractable.
For instance, in order to compute the set of legal maps that improve on the existing plan (Act 43, drawn in 2011) for the Wisconsin State Assembly, one must find all possible ways to assign 6895 precincts to 99 districts so that each of the following conditions hold simultaneously~\cite{ruth}:
\begin{itemize}
\item all districts have equal population,
\item at most 58 counties can be split in different districts,
\item at most 62 municipalities can be split,
\item the average Reock score\footnote{The Reock score~\cite{Reock:61} is the ratio of the area of the district to the area of the smallest circle containing the district.} is at least 0.39,
\item the average Polsby--Popper score\footnote{The Polsby--Popper score~\cite{PolsbyP:91} is $4\pi$ times the area of the district divided by the square of its perimeter.} is at least 0.28,
\item at least 6 districts satisfy that at least 40\% of their citizens of voting age are black, and
\item districts 8 and 9 do not change (previously ordered by a federal court).
\end{itemize}

In order to work around this apparent intractability, Markov Chain Monte Carlo simulations have been developed to generate a random ensemble of representative maps, which are then used to estimate the distribution of relevant observables~\cite{ravier, fifield, ChoL:18}.
Recent work~\cite{pegden} provides local statistical tests based on Markov Chains to prove that certain maps are unfair (given that they are outliers) with no need to provide a definition nor characterization of fair redistricting. 

In this work, we prove that computational intractability is inherent to the redistricting problem.
In particular, we show that even for simple definitions of ``fair'' and ``legal,'' deciding whether there exists a fair redistricting among legal maps is $\mathsf{NP}$-hard.
While this result is mostly relevant from a theoretical point of view (i.e., worst-case complexity says very little about real-world maps), it should help researchers gauge what sort of performance guarantees are provable with their redistricting algorithms.
In words, our definition of \textit{fair} requires that a party of interest receives at least some prescribed level of representation, while our definition of \textit{legal} only requires that
\begin{itemize}
\item all districts have approximately the same population, and
\item all districts satisfy a mild notion of geographic compactness.
\end{itemize}  
While our fairness criterion is not currently a legal standard, it has been used in court to argue that a given map is the result of a partisan gerrymander (for instance, in terms of efficiency gap).
Our result identifies a fundamental tension between judging a partisan gerrymander by the shape of voting districts (i.e., violating some notion of geographic compactness) and judging it by the impact of the gerrymander (i.e., violating a desired level of proportionality or efficiency gap); see~\cite{AlexeevM:17a,AlexeevM:17b} for additional results along these lines.

\section{Main result}  

Throughout, we denote $[n]:=\{1,\ldots,n\}$.
The goal of redistricting is to partition a state into districts that satisfy various criteria.
Let $\mathscr{D}_1\sqcup\cdots\sqcup\mathscr{D}_k=\mathbb{R}^2$ denote a partition of the plane into $k$ districts.
Suppose there are $n$ voters labeled by members of $[n]$, and suppose $\mathsf{loc}\colon[n]\rightarrow\mathbb{Q}^2$ reports the location of each voter.
Then the voters that reside in district $\mathscr{D}_i$ make up the set $D_i:=\mathsf{loc}^{-1}(\mathscr{D}_i)\subseteq[n]$.
Intuitively, we want each $D_i$ to have about the same size, since this will ensure that voters from different districts receive equal representation:
\begin{itemize}
\item[(F1)]
$(1-\gamma)\frac{n}{k}\leq |D_l| \leq (1+\gamma)\frac{n}{k}$ for every $l\in[k]$.
\end{itemize}
(Here, $\gamma\in(0,1)$ is a rational number that is chosen to be appropriately small.)
Assuming equal voter turnout, then this is equivalent to \textbf{one person, one vote}, which requires that districts contain roughly equal size populations.
A series of U.S. Supreme Court decisions in the 1960s ruled that one person, one vote must hold for all levels of redistricting~\cite{Smith:14}.

Next, we consider \textbf{geographic compactness}, which is a geometric requirement on the shape of the districts $\mathscr{D}_i$.
Indeed, gerrymandering is historically detected by districts exhibiting bizarre shape (even the etymology of ``gerrymander'' comes from likening the shape of a Massachusetts voting district to the profile of a salamander~\cite{Griffith:07}).
To enforce geographic compactness, one may force all of the districts to have a large Reock score~\cite{Reock:61} or Polsby--Popper score~\cite{PolsbyP:91}.
Another popular score is the \textbf{convex hull score}, which is the ratio $|\mathscr{D}|/|\operatorname{hull}(\mathscr{D})|$ of the area of a given district $\mathscr{D}$ to the area of its convex hull~\cite{NiemiGCH:90}.
Of course, a district receives the maximum possible convex hull score (i.e., $1$) precisely when it is convex.
Notice that a partition $D_1\sqcup\cdots\sqcup D_k=[n]$ can be realized from convex districts in the plane if and only if
\begin{itemize}
\item[(F2)]
$\operatorname{hull}(\{\mathsf{loc}(i)\}_{i\in D_l})\cap \operatorname{hull}(\{\mathsf{loc}(i)\}_{i\in D_{l'}})=\emptyset$ for every $l,l'\in[k]$ with $l\neq l'$.
\end{itemize}
Note that (F2) can be checked in polynomial time by linear programming.
While the convex hull score is mathematically convenient to work with, it can lead to undesirable districts.
For example, if we partition a state into $k$ horizontal strips, then we achieve the best possible convex hull scores for each district, and yet every district will contain voters from the far-east and west sides of the state.
This lack of compactness would be indicated by a large Reock score or Polsby--Popper score.
Alternatively, we can enforce a bound on the distance between any two voters in a given district:
\begin{itemize}
\item[(F3)]
$\|\mathsf{loc}(i)-\mathsf{loc}(j)\|_2\leq d$ whenever $i,j\in D_l$ for some $l\in [k]$.
\end{itemize}
Here, $d>0$ is a rational number.

Amazingly, it is possible to concoct an extremely partisan gerrymander, even when constrained to geographically compact districts that satisfy one person, one vote.
For example, Figure~1 in~\cite{AlexeevM:17b} illustrates two different partitions of Wisconsin into eight equally populated districts that use straight line segment boundaries; the first partition makes all eight districts majority-Republican, whereas the second makes all but one district majority-Democrat (this is the most possible since Wisconsin is majority-Republican).
To defeat partisan gerrymandering in court, one might compare voter preferences to election outcomes.
Indeed, if the proportion of seats won by a party in a state is significantly different from the proportion of votes cast for that party in that state, then one might blame partisan gerrymandering for the discrepancy.
Along these lines, arguments in the U.S. Supreme Court have leveraged the notion of proportionality (in \textit{Davis v.\ Bandemer}) and of efficiency gap (in \textit{Gill v.\ Whitford}).

In our formulation, we let $\mathsf{pref}\colon[n]\rightarrow\{0,1\}$ denote the function that reports the preference of each voter.
In practice, this function can be estimated with the help of past election data.
Then for some integer $m>0$, one can ask for the following:
\begin{itemize}
\item[(F4)]
the number of $l\in[k]$ such that $|D_l\cap\mathsf{pref}^{-1}(1)|>\frac{1}{2}|D_l|$ is at least $m$.
\end{itemize}
As an extreme example, if half of the voters have preference $1$ and $k$ is large, then it is reasonable to ask for at least $m=1$ of the districts to be majority-$1$.
(We note that the lack of symmetry between preference~$1$ and preference~$0$ in (F4) will not be important to our formulation; for example, one may replace $1$ with $0$ in (F4) in order to ensure that preference-$0$ voters also receive sufficient representation.)

\begin{definition}\
\begin{itemize}
\item[(a)]
Given $\gamma\in(0,1)$ and $d>0$, we say a partition $D_1\sqcup\cdots \sqcup D_k=[n]$ is $(\gamma,d)$-\textbf{legal} for $\mathsf{loc}\colon[n]\rightarrow\mathbb{Q}^2$ if it satisfies (F1)--(F3).
\item[(b)]
Given $m>0$, we say a partition $D_1\sqcup\cdots \sqcup D_k=[n]$ is $m$-\textbf{fair} for $\mathsf{pref}\colon[n]\rightarrow\{0,1\}$ if it satisfies (F4).
\item[(c)]
For each $\gamma\in(0,1)$ and $d>0$, the \textbf{$(\gamma,d)$-legal redistricting decision problem} takes as input $(n,k,\mathsf{loc})$ and returns whether there exists a partition $D_1\sqcup\cdots \sqcup D_k=[n]$ that is $(\gamma,d)$-legal for $\mathsf{loc}$.
\item[(d)]
For each $\gamma\in(0,1)$ and $d>0$, the \textbf{fair $(\gamma,d)$-legal redistricting decision problem} takes as input $(n,k,\mathsf{loc},m,\mathsf{pref})$ and returns whether there exists a partition $D_1\sqcup\cdots \sqcup D_k=[n]$ that is both $(\gamma,d)$-legal for $\mathsf{loc}$ and $m$-fair for $\mathsf{pref}$.
\end{itemize}
\end{definition}

In words, $(\gamma,d)$-legal redistricting asks whether there exists a legal map (satisfying (F1)--(F3)), whereas fair $(\gamma,d)$-legal redistricting asks whether there exists a fair map (satisfying (F4)) among legal maps.
Note that $(\gamma,d)$-legal redistricting amounts to a planar clustering problem that is likely $\mathsf{NP}$-hard (given its resemblance to problems in~\cite{Megiddo:84,UlrichRW:94}), but this will not play a role in our result.
What follows is our main result:

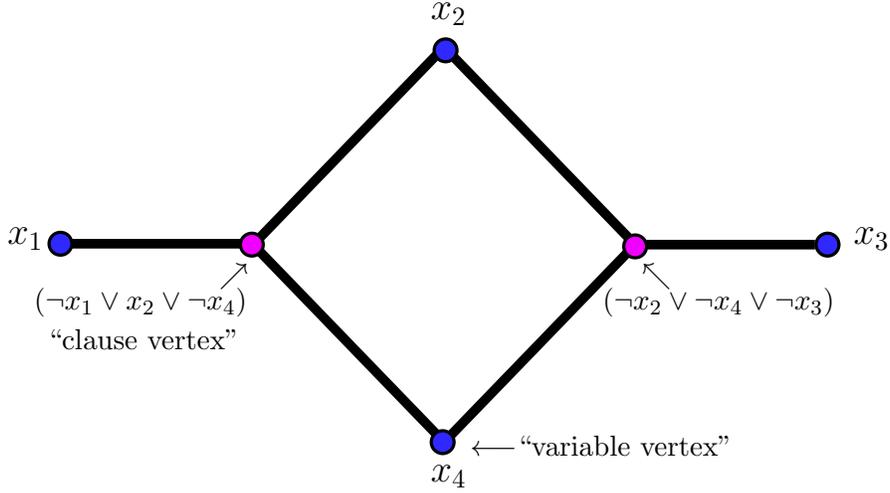
\begin{figure}
\begin{center}
\begin{picture}(300,150)
\put(0,0){\includegraphics[width=300pt]{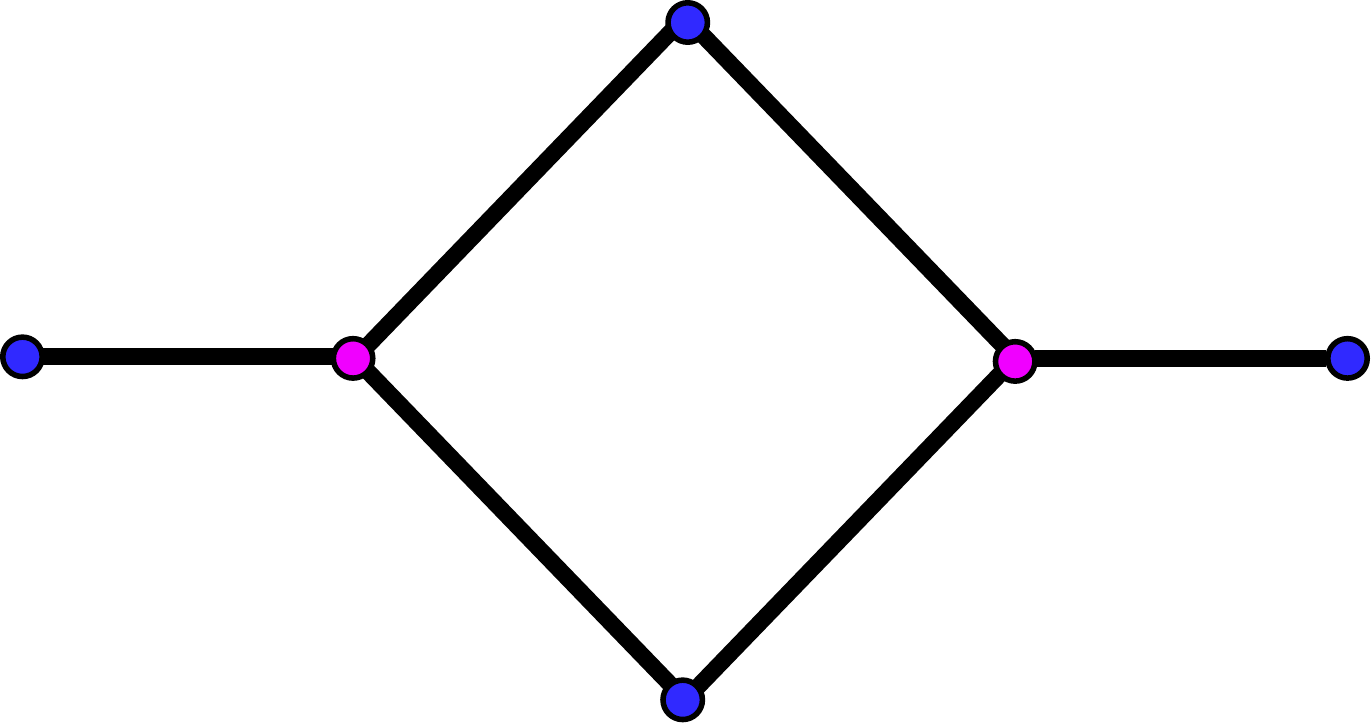}}
\put(145,165){\Large{$x_2$}}
\put(145,-10){\Large{$x_4$}}
\put(160, 0){$\longleftarrow$``variable vertex"}
\put(-15,80){\Large{$x_1$}}
\put(0,40){``clause vertex"}
\put(-5,55){$(\neg x_1 \vee x_2 \vee \neg x_4)$}
\put(65,65){$\nearrow$}
\put(210,55){$(\neg x_2 \vee \neg x_4 \vee \neg x_3)$}
\put(225,65){$\nwarrow$}
\put(305,80){\Large{$x_3$}}
\end{picture}
\end{center}
\caption{ \label{fig.planar3SAT} A planar 3-SAT instance is a special type of 3-SAT instance that corresponds to a planar graph.
The graph has one ``variable vertex" per variable, and one ``clause vertex" per clause, and two vertices share an edge precisely when one vertex is a variable vertex assigned to variable $x_j$, the other is a clause vertex assigned to clause $c_t$, and either $x_j$ or $\neg x_j$ appear in $c_t$. Here, we illustrate a planar embedding of the planar 3-SAT instance $(\neg x_1 \vee x_2 \vee \neg x_4) \wedge (\neg x_2 \vee \neg x_4 \vee \neg x_3)$.} 
\end{figure}

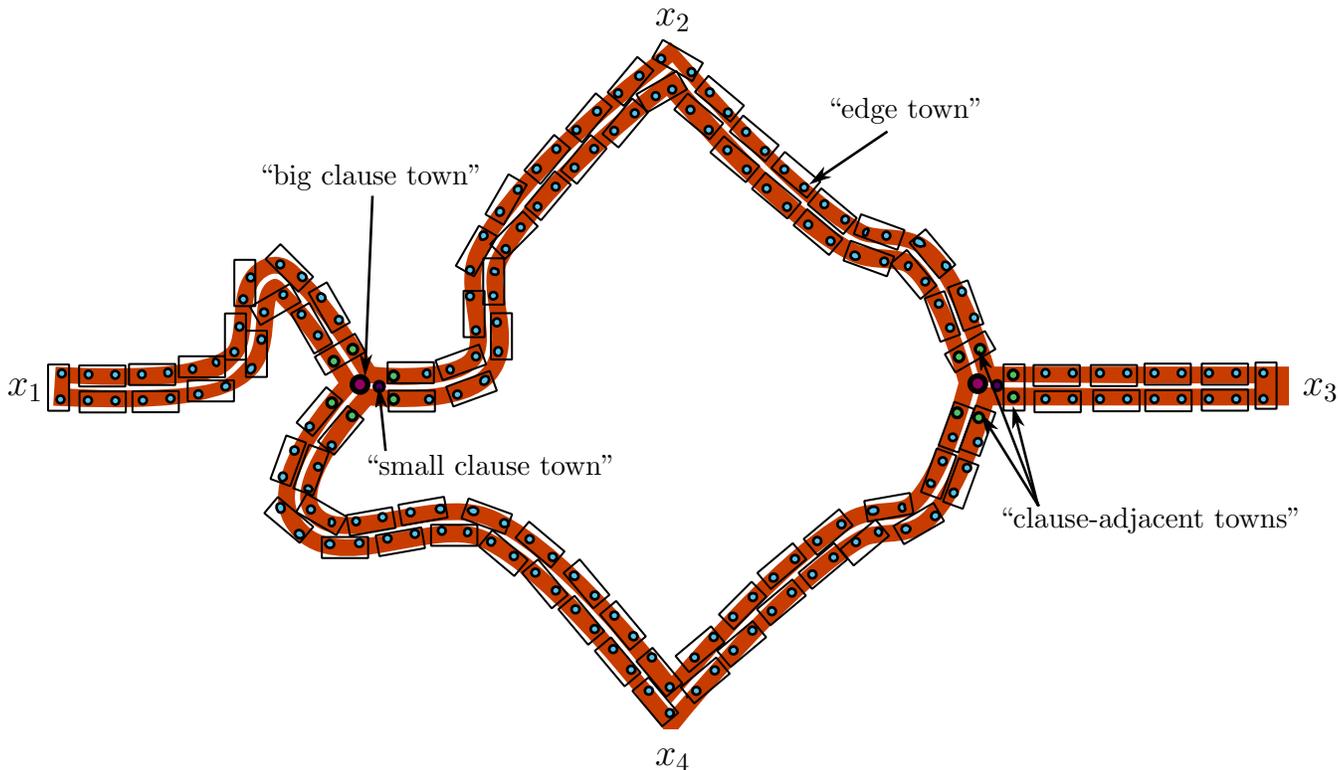
\begin{figure}
\begin{picture}(330,240)
\put(0,0){\includegraphics[width=\textwidth]{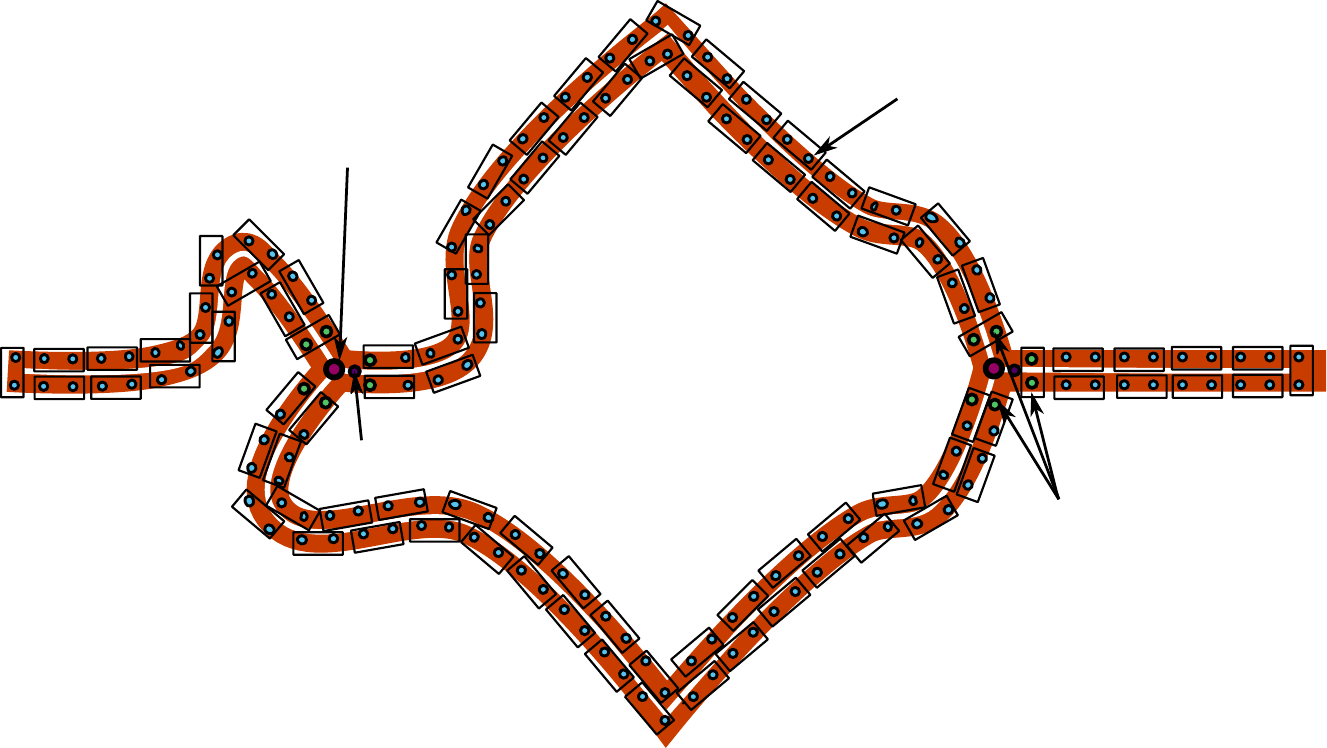}}
\put(230,270){\Large{$x_2$}}
\put(-15,130){\Large{$x_1$}}
\put(475,130){\Large{$x_3$}}
\put(230,-10){\Large{$x_4$}}
\put(360,80){``clause-adjacent towns"}
\put(295,235){``edge town"}
\put(120,100){``small clause town"}
\put(80,210){``big clause town"}
\end{picture}
\caption{\label{fig.redistricting} Our reduction converts any planar 3-SAT instance into an instance of fair $(\gamma,d)$-legal redistricting. The 3-SAT instance in this example is $(\neg x_1 \vee x_2 \vee \neg x_4) \wedge (\neg x_2 \vee \neg x_4 \vee \neg x_3)$. The districts correspond to the solution $(x_1,x_2,x_3,x_4)=(0,0,0,1)$ to the original planar 3-SAT instance. Notice that the number of edge towns between the two clauses that correspond to the variable $x_2$ ($x_4$) is odd (even) because the variable appears with opposite (equal) sign in these clauses.} 
\end{figure}

\begin{theorem}
\label{thm.main result}
For every rational $\gamma\in(0,1)$ and $d>0$, fair $(\gamma,d)$-legal redistricting is $\mathsf{NP}$-complete.
Furthermore, there exists a promise of instances over which fair $(\gamma,d)$-legal redistricting remains $\mathsf{NP}$-hard even though the corresponding instances of $(\gamma,d)$-legal redistricting can be decided in polynomial time.
\end{theorem}

In words, it is sometimes hard to find a legal redistricting that is fair, even when it is easy to find a legal redistricting.
Our proof follows by reduction from planar 3-SAT, taking inspiration from Mahanjan, Minbhorkar and Varadarajan's proof that planar $k$-means is $\mathsf{NP}$-hard~\cite{MahajanNV:12}.
Interestingly, our reduction takes $|\mathsf{pref}^{-1}(1)|>\frac{1-\gamma}{2}n$ and $m=o(k)$, meaning that in the worst case, it is hard to decide whether there exists a redistricting satisfying (F1)--(F3) that ensures that preference-1 voters win at least a vanishing fraction of the districts, even if they make up nearly half of the popular vote.
In what follows, we first describe the reduction, then we prove that the resulting instance of fair redistricting is satisfiable if and only if the original planar 3-SAT instance is satisfiable, and finally, we describe the technical details of how the reduction can be performed in polynomial time.

Given any instance of planar 3-SAT (illustrated in Figure~\ref{fig.planar3SAT}), we prescribe a corresponding instance of fair $(\gamma,d)$-legal redistricting (illustrated in Figure~\ref{fig.redistricting}).
We are given a planar bipartite graph of clause vertices and variable vertices, where each clause vertex has degree 3.
Construct a planar embedding of this graph with integer vertex coordinates.
At each clause vertex, locally deform the incident edges so that they approach the clause vertex at $0$, $\frac{2\pi}{3}$ and $\frac{4\pi}{3}$ radians.
Next, we place ``towns'' at various points in the plane relative to this graph embedding before placing voters at each town.
Pick $\eta>0$ and $\epsilon>0$ to be sufficiently small rational numbers (we will have $\epsilon\ll\eta$, and both will be polynomially small in the number of vertices).
Place a ``big clause town'' at each clause vertex and a ``small clause town'' $\epsilon$ away at $0$ radians.
About each clause vertex, place six ``clause-adjacent towns'' at distances between $0.99\eta$ and $\eta$ away and at angles approximately equal to $\frac{2\pi j}{18}$ for $j\in\{1,5,7,11,13,17\}$.
For each variable vertex $x$, consider the set $C_x\subseteq\mathbb{R}^2$ of points that are between $\eta$ and $\frac{3}{2}\eta$ away from the edges incident to $x$.
Place ``edge towns'' throughout $C_x$ so that (1) each edge town has exactly two towns that are between $0.99\eta$ and $\eta$ away (one of which might be clause-adjacent) and none closer, and (2) the number of edge towns between consecutive clause-adjacent towns is even (odd) when the signs of $x$ in these clauses are equal (opposite).

As we verify later, all of these town locations can be selected to have coordinates with a logarithmically small number of digits of precision.
Before placing voters at the various towns, we first multiply all of the towns' coordinates by the rational number $d/(\eta+\epsilon)$.
After this multiplication, the description lengths of these coordinates with remain logarithmically small.
For simplicity of exposition, we assume $d=\eta+\epsilon$ without loss of generality so that no multiplication is necessary.

We now place voters at the various towns.
Take $L>1/\gamma^2$ to be a multiple of $4$.
Each big and small clause town receives $L$ and $\lfloor \frac{2\gamma}{3} L\rfloor$ voters, respectively, all of which have preference 1.
Each clause-adjacent town receives $\frac{L}{2}+\lfloor \frac{\gamma}{6} L\rfloor$ voters, $\frac{L}{4}$ with preference 1 and $\frac{L}{4}+\lfloor \frac{\gamma}{6} L\rfloor$ with preference 0.
In each edge town, we place $\frac{L}{2}$ voters, $\frac{L}{4}-\lfloor \frac{\gamma}{4} L\rfloor$ with preference 1 and $\frac{L}{4}+\lfloor \frac{\gamma}{4} L\rfloor$ with preference 0.
Take $k$ to be the number of clauses plus half the number of non-clause towns, and set $m$ to be twice the number of clauses.

\begin{lemma}
Given any instance of planar 3-SAT, the corresponding instance of fair $(\gamma,d)$-legal redistricting is satisfiable if and only if the original instance of planar 3-SAT is satisfiable.
\end{lemma}

\begin{proof}
In order to satisfy (F2), all voters in a given town belong to the same district.
Next, (F1) and (F3) together force every edge town to be matched with one of the two towns that are approximately $\eta$ away.
Each clause-adjacent town is either matched with the edge town approximately $\eta$ away, or the nearest clause-adjacent town.
In the latter case, the corresponding small clause town may also join the district, but it may not join any other clause-adjacent town in order to maintain (F3).
The big clause towns may be matched with the corresponding small clause town when it is not matched with clause-adjacent towns; otherwise, these big clause towns are so big that they form their own districts in order  to satisfy (F1).

Overall, for each variable $x$, the corresponding edge towns and clause-adjacent towns are perfectly matched in one of two ways (namely, one of two perfect matchings in an even cyclic graph).
For a clause that includes $x$ (the negation of $x$), we may interpret the corresponding clause-adjacent towns as sharing a district precisely when $x$ is true (false).
As such, the underlying instance of planar 3-SAT is satisfiable precisely when there exist districts satisfying (F1)--(F3) such that, for every clause, there exists a corresponding pair of clause-adjacent towns that share a district (in which case, the small clause town may join their district).

The majority-1 districts are precisely the ones that contain a big or small clause town, and so there are at most $m$ such districts.
Equality occurs, namely (F4), precisely when each of the small clause towns is matched with some pair of clause-adjacent towns, which is feasible precisely when the underlying instance of planar 3-SAT is satisfiable.
\end{proof}

Next, we quickly verify that given any instance of planar 3-SAT, the corresponding instance of $(\gamma,d)$-legal redistricting can be solved in polynomial time.
(Then the promise of instances in Theorem~\ref{thm.main result} can be taken to be the image of our reduction from planar 3-SAT.)
First, assign each town that contains $L$ voters to its own district, and then add as many voters to these districts as possible while satisfying (F1) and (F3).
After doing so, there will be $m/2$ districts, each containing a big clause town and the corresponding small clause town.
Next, define a graph such that the vertices are the remaining towns, with two towns being adjacent if their distance is at most $\eta$.
This graph is a disjoint union of cycles, and one may partition the remaining towns into districts by selecting any perfect matching of towns.
The result is a $(\gamma,d)$-legal redistricting that can be computed in polynomial time.

To prove Theorem~\ref{thm.main result}, it remains to show that our reduction from planar 3-SAT can be accomplished in polynomial time.
To this end, there are three nontrivial subroutines to analyze:
\begin{itemize}
\item[(i)]
Given an instance of planar 3-SAT, find a corresponding planar embedding.
\item[(ii)]
Given a planar 3-SAT embedding, locally deform the edges incident to each clause vertex.
\item[(iii)]
Place towns throughout the plane with the appropriate geometry and parity.
\end{itemize}
For (i), we appeal to the following:

\begin{proposition}[Chrobrak--Payne~\cite{chrobak1995linear}]
\label{prop.embedding}
There exists a linear-time algorithm that, given an $N$-vertex planar graph as input, outputs a planar embedding of that graph such that
\begin{itemize}
\item[(a)]
the embedded vertices lie in a $2N\times 2N$ integer grid, and
\item[(b)]
the embedded edges are line segments.
\end{itemize}
\end{proposition}

We will exploit the form of this embedding in our analysis of (ii) and (iii).
To accomplish (ii), we first identify a suitably small neighborhood of each clause vertex $v$, specifically, the set of points $B(v,\delta_1)$ in the plane that are within some $\delta_1>0$ of $v$ in $\infty$-norm.
Provided $\delta_1$ is small enough, then the portion of the Chrobrak--Payne embedding that resides in $B(v,\delta_1)$ amounts to three segments emanating from $v$, and the angle between any two of these segments is at least $\delta_2$ radians (for some appropriately small $\delta_2>0$).

For simplicity, we may assume that the angles of these segments are at least $\delta_2$ away from $T:=\{0,\frac{2\pi}{3},\frac{4\pi}{3}\}\cup\{\frac{\pi}{8},\frac{3\pi}{8},\frac{5\pi}{8},\frac{7\pi}{8}\}$; indeed, if this fails to hold, we can redefine $\delta_2\leftarrow \delta_2/4$, and if there is still a segment of angle $\theta$ within $\delta_2$ of $T$, then we can modify that segment to be a polygonal curve so that the portion in $B(v,\delta_1/2)$ is a segment with angle $\theta+2\delta_2$, and then redefine $\delta_1\leftarrow \delta_1/2$.
Since the angles avoid $\{0,\frac{2\pi}{3},\frac{4\pi}{3}\}$, they each reside in one of three sections: $(0,\frac{2\pi}{3})$, $(\frac{2\pi}{3},\frac{4\pi}{3})$, $(\frac{4\pi}{3},2\pi)$.
As such, there are three cases to consider: (1) each angle resides in a different section, (2) one section contains exactly two angles, and (3) one section contains all three angles.
By rotating and reflecting as necessary, one can ensure that $(0,\frac{2\pi}{3})$ receives at least as many angles as $(\frac{4\pi}{3},2\pi)$, which receives at least as many angles as $(\frac{2\pi}{3},\frac{4\pi}{3})$.

\begin{figure}
\includegraphics[width=0.3\textwidth]{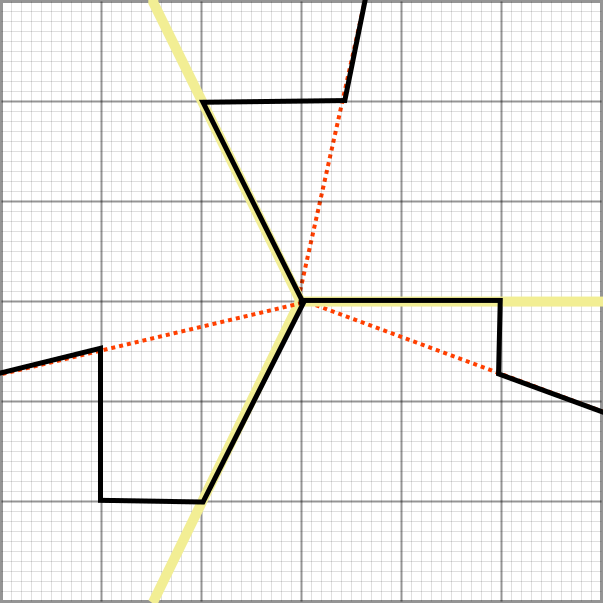}
\qquad
\includegraphics[width=0.3\textwidth]{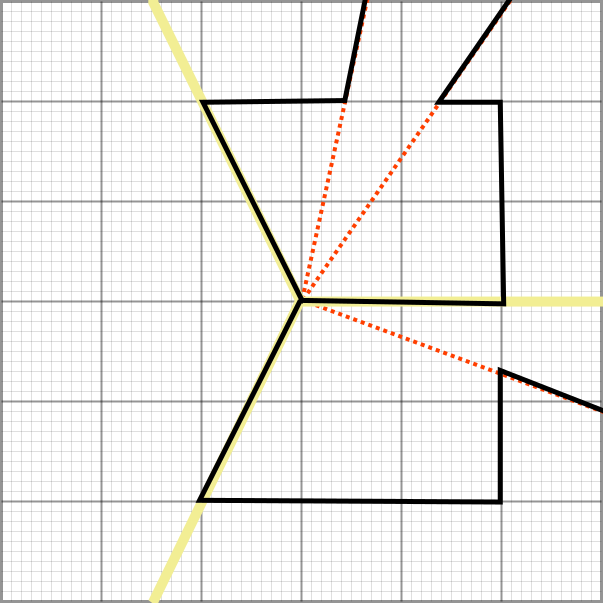}
\qquad
\includegraphics[width=0.3\textwidth]{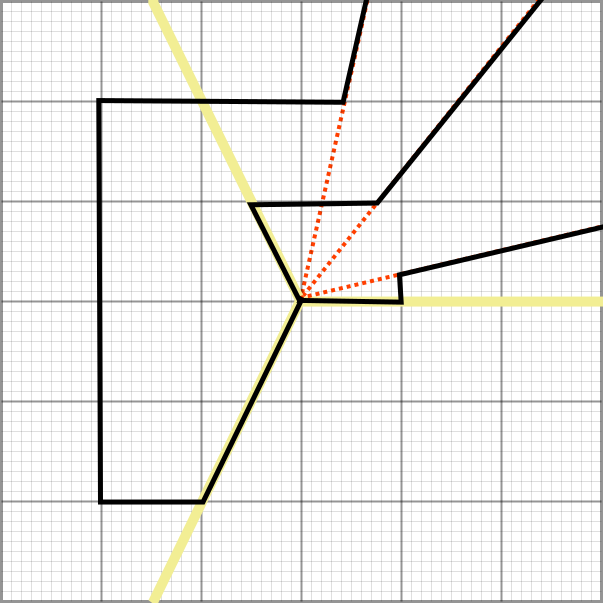}
\caption{\label{fig.redirection} Deformation of the Chrobrak--Payne embedding of a planar 3-SAT instance in an $\infty$-norm ball centered at a clause vertex with radius $\delta_1$. In each case, the clause vertex has three edges (dotted) that are redirected along the boundary of the $\infty$-norm balls of radius $\frac{\delta_1}{3}$ and/or $\frac{2\delta_1}{3}$.}
\end{figure}

With this standardized form, Figure~\ref{fig.redirection} illustrates how to modify the segments into polygonal curves in each case so that the portions in $B(v,\delta_1/3)$ are segments with angles $\{0,\frac{2\pi}{3},\frac{4\pi}{3}\}$.
Since the original segment angles were at least $\delta_2$ from $T$, the segments of the new polygonal curves are not too small.
In particular, this modification to the Chrobrak--Payne embedding has the property that each edge is embedded as a polygonal curve of segments such that
\begin{itemize}
\item[(D1)]
every segment has length at least $\delta:=\frac{1}{100}\min\{\delta_1,\delta_2\}$, and
\item[(D2)]
every pair of disjoint segments in the graph embedding has distance at least $\delta$.
\end{itemize}
Importantly, $\delta_1$ and $\delta_2$ (and therefore $\delta$) are polynomially small, which follows from Proposition~\ref{prop.embedding} and the following lemma:

\begin{lemma}
\label{lem.grid bounds}
Let $A$, $B$ and $C$ be distinct points in an $\ell\times \ell$ integer grid.
Then
\begin{itemize}
\item[(a)]
either $\sin(\angle ABC)=0$ or $\sin(\angle ABC)\geq\frac{1}{2\ell^2}$, and
\item[(b)]
either $\operatorname{dist}(A,\overline{BC})=0$ or $\operatorname{dist}(A,\overline{BC})\geq\frac{1}{2\ell^2}$.
\end{itemize}
\end{lemma}

\begin{proof}
Put $A,B,C\in[\ell]^2$ and denote $(a,b)=A-B$, $(c,d)=C-B$.
Then
\[
|\det(\begin{smallmatrix}a&b\\c&d\end{smallmatrix})|
=\|(a,b)\|_2\cdot \|(c,d)\|_2\cdot \sin(\angle ABC).
\]
Assuming $\sin(\angle ABC)\neq0$, then by integrality, the left-hand side is at least $1$, whereas the right-hand side is at most $2\ell^2\sin(\angle ABC)$.
Rearranging then gives (a).

For (b), let $D$ denote the closest point in $\overline{BC}$ to $A$.
If $D=A$, then the distance is zero.
If $D=B$ or $C$, then the distance is at least $1$ since in this case, $A$ and $D$ are distinct points in an integer grid.
Otherwise, $D\neq A$ is an interior point of $\overline{BC}$, and by the Hilbert projection theorem, $\angle ADB$ is a right angle.
Since $A$ and $B$ are distinct, we have $\operatorname{dist}(A,B)\geq 1$, and so
\[
\operatorname{dist}(A,\overline{BC})
=\operatorname{dist}(A,D)
\geq\frac{\operatorname{dist}(A,D)}{\operatorname{dist}(A,B)}
=\sin(\angle ABD)
=\sin(\angle ABC)
\geq\frac{1}{2\ell^2},
\]
where the last step is by (a).
This gives (b).
\end{proof}

Finally, we analyze (iii).
Let $N$ denote the number of vertices in the original planar 3-SAT instance, compute an embedding of the corresponding planar graph in a $2N\times 2N$ grid (this is possible by Propostion~\ref{prop.embedding}), and locally deform the edges incident to each clause vertex so as to satisfy (D1) and (D2) with $\delta=1/(10^4N^2)$ (this is possible by Lemma~\ref{lem.grid bounds}).
For each of the (polynomially many) segments of this new embedding, store the coordinates of each endpoint with $p:=\lceil 2\log_{10} N\rceil+10$ digits of precision.

We now describe where to place the various towns in our instance of fair redistricting.
Select $\eta\in(0,\delta/200)$ and $\epsilon\in(0,\eta/200)$, each with $p$ digits of precision.
It is helpful to define a sequence $\{z_j\}_{j\in[t]}$ of $t=1800$ points; take $z_j$ to have coordinates $(\eta\cos(2\pi j/t),\eta\sin(2\pi j/t))$, but rounded to $p+20$ digits of precision so as to have norm between $0.99\eta$ and $\eta$.
Now for each clause vertex $v\in\mathbb{Z}^2$, place a big clause town at $v$, a small clause town at $v+(\epsilon,0)$, and place clause-adjacent towns at $v+z_j$ for $j\in\{100,500,700,1100,1300,1700\}$.
Next, we will define edge town locations by iteratively adding different choices of $z_j$ to existing town locations; see Figure~\ref{fig.detailed} for an illustration.

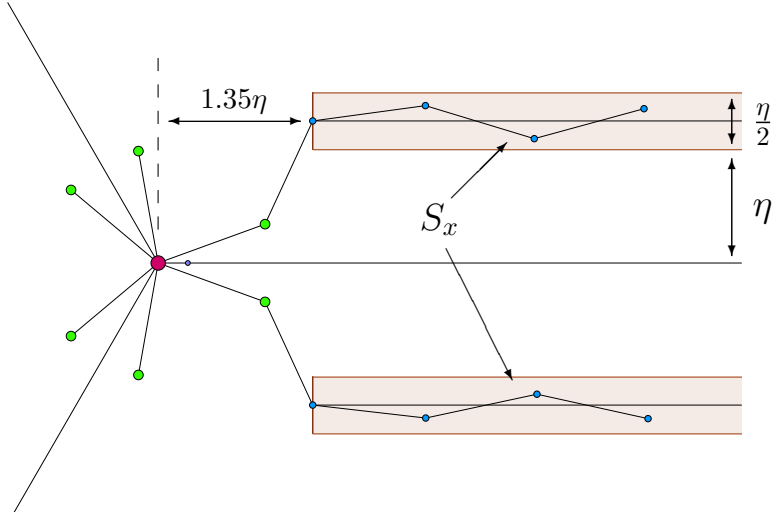
\begin{figure}
\begin{center}
\begin{picture}(330,200)
\put(0,0){\includegraphics[width=0.7\textwidth]{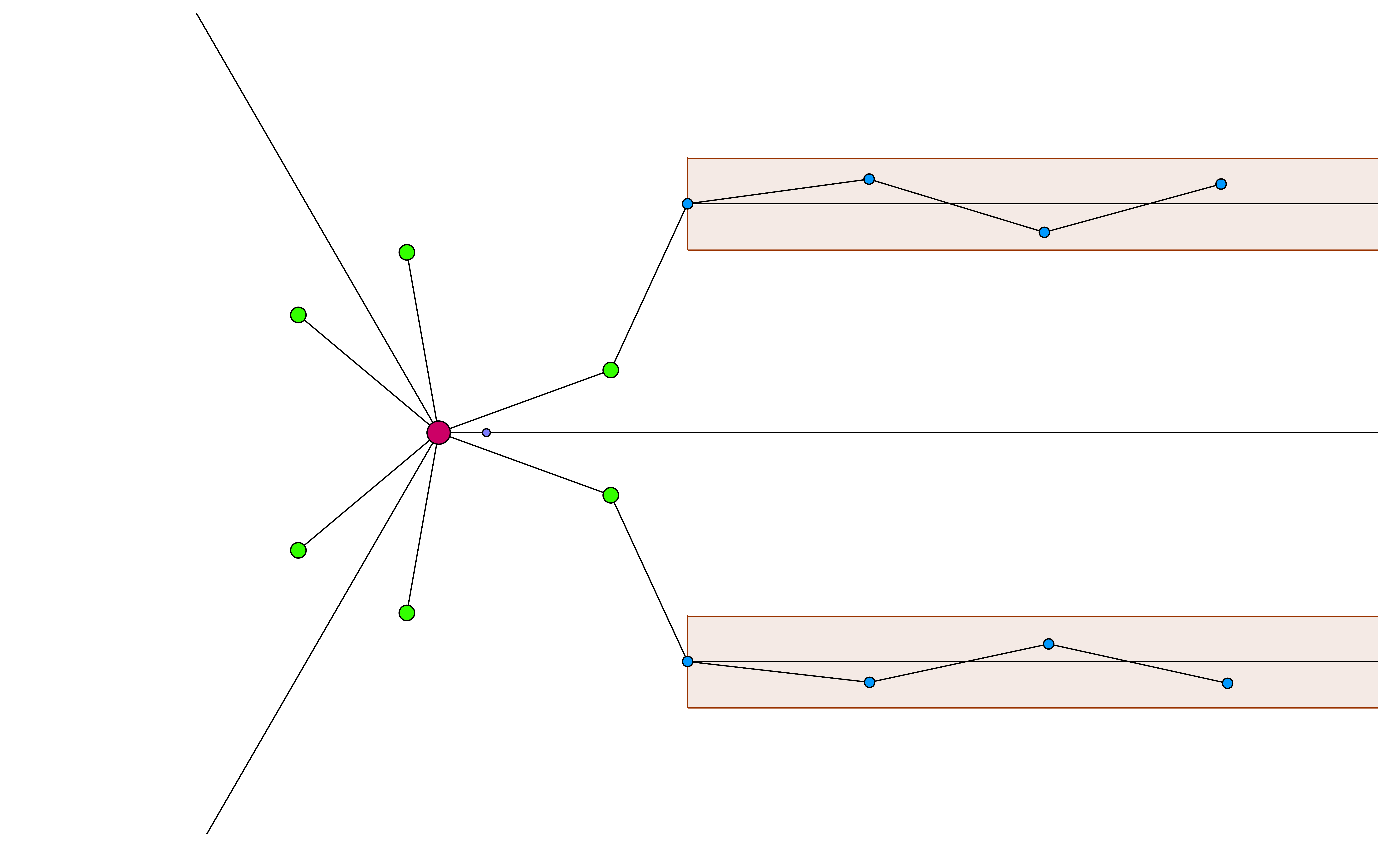}}
\put(328,115){\Large{$\eta$}}
\put(328,147){\Large{$\frac{\eta}{2}$}}
\put(119,156){$1.35\eta$}
\put(202,110){\Large{$S_x$}}
\put(320,100){\vector(0,1){37}}
\put(320,137){\vector(0,-1){37}}
\put(320,142.5){\vector(0,1){17}}
\put(320,159.5){\vector(0,-1){17}}
\put(107,151){\vector(1,0){50}}
\put(157,151){\vector(-1,0){50}}
\multiput(103,110)(0,10){7}{\line(0,10){5}}
\put(214,122){\vector(1,1){21}}
\put(212,103){\vector(1,-2){25}}
\end{picture}
\end{center}
\caption{\label{fig.detailed} Detailed illustration of towns near a clause vertex. 
At this scale, edges incident to the clause vertex amount to segments at angles $0$, $\frac{2\pi}{3}$ and $\frac{4\pi}{3}$.
Let $x$ denote the variable vertex (not depicted) that is incident to the edge at angle $0$.
Big and small clause towns are plotted in magenta and violet, whereas clause-adjacent towns are distributed about the big clause town in green. Two different components of $S_x$ are found on both sides of an edge at distance $\eta$ away. Edge towns (depicted in blue) are placed inside each of these components.}
\end{figure}

Recall that for each variable vertex $x$, $C_x$ denotes the set of points that are between $\eta$ and $\frac{3}{2}\eta$ away from the edges incident to $x$.
Let $S_x$ denote the subset obtained by removing any points from $C_x$ whose nearest point in the edges incident to $x$ is within $1.35\eta$ of a clause vertex.
In words, $S_x$ is a ``broken'' version of $C_x$; while $C_x$ is connected, the number of connected components of $S_x$ equals the number of clause vertices adjacent to $x$ in the original graph.
Furthermore, since
\[
1.35
<\cos\tfrac{\pi}{9}+\sqrt{1-(\tfrac{5}{4}-\sin\tfrac{\pi}{9})^2},
\]
then for each clause-adjacent vertex $u$ associated with an edge $e$ incident to $x$, there exists $j\in[t]$ such that $u+z_t$ lies in $S_x$ and has distance about $\frac{5}{4}\eta$ from $e$, and we place an edge town at $u+z_t\in S_x$.
To clarify, each component of $S_x$ can be thought of as an $\frac{\eta}{4}$-thickened curve, and we currently have edge towns at both ends of this curve (modulo precision).
Furthermore, it is straightforward to show that every pair of edge towns currently has distance greater than $1.08\eta>\eta+\epsilon$.
It remains to place edge towns throughout the remainder of each component.

Fix a component of some $S_x$, and let $w$ and $w'$ denote the locations of the two edge towns that have been placed in this component.
Put $w_0=w$, and given $w_i$, find $j$ such that $w_i+z_j\in S_x$ is more than $\eta+\epsilon$ away from all towns other than $w_i$ while being as close as possible to $\frac{5}{4}\eta$ away from the edge set of the embedded graph, and put an edge town at $w_{i+1}=w_i+z_j$.
(This computation is feasible since the embedded graph is comprised of polynomially many segments, and there are always polynomially many towns.)
We terminate this iteration once $w_i$ satisfies $\|w_i-w'\|_2\leq 100\eta$ and $i$ has the appropriate parity (after this iteration, we will add an odd number of edge towns to this component of $S_x$).

At this point, $\eta<\delta/200$ and (D1) together imply that the portion of $S_x$ between $w_i$ and $w'$ corresponds to a straight line segment of the edge set of the graph embedding.
For simplicity of exposition, we will assume $w_i=(0,0)$ and that $w'=(x',0)$ lies on the positive $x$-axis, with the understanding that we may rotate and translate as necessary before rounding to $p+20$ digits of precision.
With this orientation, we have $[0,x']\times[-\frac{\eta}{5},\frac{\eta}{5}]\subseteq S_x$.
Taking $r:=\lfloor x'/1.95\rfloor$, one may use the fact that $x'\geq 98\eta$ to verify that $q:=x'/r\in[1.95\eta,1.99\eta]$.
Then we place an edge town at $(jq,0)$ for every $j\in[r-1]$ and at $((j-\frac{1}{2})q,\sqrt{(0.995\eta)^2-(q/2)^2})$ for every $j\in[r]$ (the second coordinate lies in $[0,\frac{\eta}{5}]$ since $q\in[1.95\eta,1.99\eta]$).
Overall, for each component of each $S_x$, we have identified locations for edge towns in polynomial time, completing the construction of a fair $(\gamma,d)$-legal redistricting instance.

\section*{Acknowledgments}
The main ideas in this paper were conceived while the authors were attending an Oberwolfach workshop on ``Applied Harmonic Analysis and Data Processing'' in March 2018.
The authors thank Boris Alexeev and Ruth Greenwood for reading a preliminary version of this paper and providing helpful comments.
RK was supported in part by Joel A.\ Tropp under ONR Award No.\ N-00014-17-12146 and also acknowledges funding provided by the Institute of Quantum Information and Matter, an NSF Physics Frontiers Center (NSF Grant PHY-1733907).
DGM was partially supported by AFOSR F4FGA06060J007 and AFOSR Young Investigator Research Program award F4FGA06088J001.
SV was partially supported by the Simons Algorithms and Geometry (A$\&$G) Think Tank. 
The views expressed in this article are those of the authors and do not reflect the official policy or position of the United States Air Force, Department of Defense, or the U.S.\ Government.


\begin{thebibliography}{WW}

\bibitem{AlexeevM:17a}
B.\ Alexeev, D.\ G.\ Mixon,
An impossibility theorem for gerrymandering,
arXiv preprint arXiv:1710.04193 (2017).

\bibitem{AlexeevM:17b}
B.\ Alexeev, D.\ G.\ Mixon,
Partisan gerrymandering with geographically compact districts,
arXiv preprint arXiv:1712.05390 (2017).

\bibitem{ravier}
S. Bangia, C. Vaughn Graves, G. Herschlag, H. S. Kang, J. Luo, J. C. Mattingly,  R. Ravier, 
Redistricting: Drawing the Line,
arXiv preprint arXiv:1704.03360 (2017).

\bibitem{duchin}
M. Bernstein, M. Duchin, 
A formula goes to court: Partisan Gerrymandering and the efficiency gap,
arXiv preprint arXiv:1705.10812 (2017).

\bibitem{pegden}
M.\ Chikina, A.\ Frieze, W.\ Pegden, 
Assessing significance in a Markov chain without mixing,
 Proceedings of the National Academy of Sciences 114, no. 11 (2017) 2860--2864.

\bibitem{ChoL:18}
W.\ K.\ Tam Cho, Y.\ Y.\ Liu,
Sampling from complicated and unknown distributions:\ Monte Carlo and Markov chain Monte Carlo methods for redistricting,
Physica A 506 (2018) 170--178.

\bibitem{chrobak1995linear}
M.\ Chrobak, T.\ H. Payne, 
A linear-time algorithm for drawing a planar graph on a grid,
Information Processing Letters 54 no. 4 (1995) 241--246.


\bibitem{fifield}
B. Fifield, M. Higgins, K. Imai, A. Tarr,
A new automated redistricting simulator using markov chain monte carlo,
Work. Pap., Princeton Univ., Princeton, NJ (2015).

\bibitem{DeFraysseixPP:90}
H.\ De Fraysseix, J.\ Pach, A.\ Pollack,
How to draw a planar graph on a grid,
Combinatorica 10 (1990) 41--51.

\bibitem{GelmanK:94}
A.\ Gelman, G.\ King,
A Unified Method of Evaluating Electoral Systems and Redistricting Plans,
Am.\ J.\ Pol.\ Sci.\ 38 (1994) 514--554.

\bibitem{ruth}
R.\ Greenwood,
Campaign Legal Center,
personal communication, April 2018. 

\bibitem{Griffith:07}
E.\ C.\ Griffith,
The Rise and Development of the Gerrymander, Scott, Foresman, 1907.

\bibitem{GrofmanK:07}
B.\ Grofman, G.\ King,
The future of partisan symmetry as a judicial test for partisan gerrymandering after LULAC v.\ Perry,
Election Law J.\ 6 (2007) 2--35.

\bibitem{mattingly_NC}
G. Herschlag, H. S. Kang, J. Luo, C. Vaughn Graves, S. Bangia, R. Ravier,  J. C. Mattingly,
Quantifying Gerrymandering in North Carolina,
arXiv preprint arXiv:1801.03783 (2018).

\bibitem{mattingly_WI}
G.\ Herschlag, R.\ Ravier, J.C.\ Mattingly,
Evaluating Partisan Gerrymandering in Wisconsin,
arXiv preprint arXiv:1709.01596 (2017).

\bibitem{MahajanNV:12}
M.\ Mahajan, P.\ Nimbhorkar, K.\ Varadarajan,
The planar k-means problem is NP-hard,
Theoretical Computer Science 442 (2012) 13--21.

\bibitem{Megiddo:84}
M.\ Megiddo, K.\ J.\ Supowit,
On the complexity of some common geometric location problems.
SIAM Journal on Computing 13, 1 (1984): 182--196.

\bibitem{NiemiGCH:90}
R.\ G.\ Niemi, B.\ Grofman, C.\ Carlucci, T.\ Hofeller,
Measuring compactness and the role of a compactness standard in a test for partisan and racial gerrymandering,
J.\ Politics 52 (1990) 1155--1181.

\bibitem{PolsbyP:91}
D.\ D.\ Polsby, R.\ D.\ Popper,
The Third Criterion:\ Compactness as a Procedural Safeguard against Partisan Gerrymandering,
Yale Law Policy Rev.\ 9 (1991) 301--353.

\bibitem{Reock:61}
E.\ C.\ Reock,
Measuring compactness as a requirement of legislative apportionment, Midwest J.\ Political Sci.\ 5 (1961) 70--74.

\bibitem{Smith:14}
J.\ D.\ Smith,
On Democracy's Doorstep:\ The Inside Story of how the Supreme Court Brought ``One Person, One Vote'' to the United States.
Hill and Wang, 2014.

\bibitem{efficiency}
N. O. Stephanopoulos, E. M. McGhee,
Partisan gerrymandering and the efficiency gap,
The University of Chicago Law Review (2015) 831--900.

\bibitem{UlrichRW:94}
U.\ Pferschy, R.\ Rudolf, G.\ J.\ Woeginger,
Some geometric clustering problems,
Nord.\ J.\ Comput.\ 1, no. 2 (1994): 246--263.

\end{thebibliography}
\end{document}